\documentclass[11pt]{article}
\usepackage[margin=1in]{geometry}

\usepackage{amsfonts}
\usepackage{amsmath}
\usepackage{amsthm} 
\usepackage{epsfig}
\usepackage{url}
\usepackage{bm}
\usepackage[colorlinks,pagebackref=true]{hyperref}

\newcommand{\e}{{\rm e}}
\newcommand{\E}{{\mathbb E}}
\newcommand{\Pa}{{\mathbb P}}
\newcommand{\Q}{{\mathbb Q}}

\newcommand{\R}{{\mathbb R}}

\newcommand{\Acal}{{\mathcal A}}

\newcommand{\Fcal}{{\mathcal F}}

\newcommand{\Hcal}{{\mathcal H}}

\newcommand{\Zcal}{{\mathcal Z}}

\DeclareMathOperator{\trace}{trace}

\newtheorem{proposition}{Proposition}[section]
\newtheorem{lemma}[proposition]{Lemma}
\newtheorem{theorem}[proposition]{Theorem}

\newtheorem{remark}[proposition]{Remark}
\newtheorem{exampleemph}[proposition]{Example}   

\newenvironment{example}{\begin{exampleemph}\begin{upshape}}{\end{upshape}\end{exampleemph}} 

\usepackage{amsmath,amsthm,amssymb,amsfonts,graphicx,dsfont}
\usepackage{rotating}
\usepackage[capposition=top]{floatrow}
\usepackage[maxfloats=100]{morefloats}
\usepackage{listings}
\usepackage{booktabs}
\usepackage{verbatim}
\usepackage{xcolor}
\usepackage[title]{appendix}
\usepackage{multirow}
\usepackage{longtable}
\usepackage{subcaption}
\usepackage{changepage}   
\makeatletter
\newcommand*{\rom}[1]{\expandafter\@slowromancap\romannumeral #1@}
\makeatother

\usepackage[utf8]{inputenc}

\begin{document}

\title{Discount Models\thanks{I thank Bent Jesper Christensen, Martin Larsson, and two anonymous referees for helpful comments.}}

\author{Damir Filipovi\'c\thanks{EPFL and Swiss Finance Institute, 1015 Lausanne, Switzerland. Email: damir.filipovic@epfl.ch}   }
\date{25 July 2023\\ forthcoming in \emph{Finance and Stochastics}}
\maketitle
\begin{abstract}
Discount is the difference between the face value of a bond and its present value. I propose an arbitrage-free dynamic framework for discount models, which provides an alternative to the Heath--Jarrow--Morton framework for forward rates. I derive general consistency conditions for factor models, and discuss affine term structure models in particular. There are several open problems, and I outline possible directions for further research.\\

\noindent\textbf{Keywords:} discount, factor models, stochastic partial differential equation, term structure models, zero-coupon bonds\\

\noindent\textbf{JEL classification:} C32, G12, G13\\

\noindent\textbf{MSC 2020 classification:} 91B70, 91G20, 91G30

\end{abstract}

\section{Discount}

Let $P(t,T)$ denote the time-$t$ price of a zero-coupon bond with maturity $T$, or short, a $T$-bond. Define the corresponding \emph{discount}
\[ H(t,T):=1-P(t,T).\]
The discount $H(t,T)$ is the difference between the face value of the bond and its present value. It is the interest earned on investing in a $T$-bond at $t$ and hold it to maturity $T$. As such, it quantifies the time value of money. It also equals the time-$t$ price of a long position in a floating rate note paying overnight short rates $r_t=-\partial_T P(t,T)|_{T=t}$ minus a short position in a $T$-bond. We call this long/short portfolio the \emph{$T$-discount}. The net cash flow generated by the $T$-discount is $r_t\,dt$ at any $t\le T$ and $0$ after $T$. The payments of the principals of the floating rate note and the zero-coupon bond at $T$ offset each other.  A $T$-discount is therefore identical to the floating leg of an overnight indexed swap with maturity $T$.

The gains process, say $G(t,T)$, from holding a $T$-discount over $[0,t]$, where the cash flows are continuously invested in the money market account that earns interest at the short rate, is given by the sum
\begin{equation}\label{gainsdef}
  G(t,T) = \underbrace{\int_0^t  \e^{\int_s^t r_u\,du}  r_s\, ds}_{\text{accumulated cash flow}} + \underbrace{H(t,T)}_{\text{spot value}}  = \e^{\int_0^t r_s\,ds} - P(t,T).
\end{equation}

%
%
%
%
%
%
%
%

\section{Discount framework}\label{secTVMM}

In analogy to the Heath--Jarrow--Morton (HJM) approach \cite{hea_jar_mor_92} for modeling the forward rates $f(t,T)=-\partial_T \log P(t,T)$, we now formulate an arbitrage-free dynamic framework for discount models. We assume a filtered probability space $(\Omega,\Fcal,(\Fcal_t)_{t\ge 0},\Q)$, satisfying the usual conditions, carrying a $n$-dimensional standard Brownian motion $W$. For simplicity of exposition, we skip technical details and assume throughout that all stochastic processes are adapted, and regular enough such that the differential and integral operations are well defined. For more background and technical details, also on arbitrage pricing, we refer to the textbooks \cite[Chapter 19]{bjo_20} and \cite{fil_09}.

We recall that $\Q$ is a risk-neutral \emph{(local) martingale measure} for the bond market if all discounted $T$-bond price processes $\e^{-\int_0^t r_s\, ds} P(t,T)$ are $\Q$-(local) martingales. In view of \eqref{gainsdef}, this holds if and only if all discounted $T$-discount gains processes $\e^{-\int_0^t r_s\, ds} G(t,T)$ are $\Q$-(local) martingales. It is well known that the bond market is arbitrage-free if and (essentially) only if $\Q$ is a local martingale measure, see, e.g., \cite[Chapter 11]{bjo_20} or \cite[Section 4.3.4]{fil_09}.

We represent the $T$-discount price at any $t\le T$ in terms of its maturity derivative,
\[  H(t,T)= \int_t^T h(t,s)\,ds ,\]
where the \emph{discount derivative} $h(t,T)$ is assumed to be an It\^o process with dynamics of the form
\begin{equation}\label{TVMprior}
 dh(t,T) = \alpha(t,T)\,dt + \sigma(t,T)\,dW_t ,
\end{equation}
for some drift and volatility processes $\alpha(t,T)$ and $\sigma(t,T)$, respectively. We then specify the $T$-bond price by
\begin{equation}\label{Tbondprior}
P(t,T) = 1 -\int_t^T h(t,s)\,ds.
\end{equation}
The $T$-bond specification \eqref{Tbondprior} is the linearized counterpart to the relation $P(t,T)=\e^{-\int_t^T f(t,s)\,ds}$. Differentiating in $T$, the relationship between the discount derivative and forward rate is obtained as
\[h(t,T)= P(t,T) f(t,T)  .\]
In particular,
\begin{equation}\label{eqSR}
  h(t,t) =  f(t,t) = r_t
\end{equation}
equals the short rate.

The challenge of our linear approach is that $T$-bond prices should be positive, $P(t,T)>0$, which is equivalent to requiring that
\begin{equation}\label{inthle1}
 \int_t^T h(t,s)\,ds<1,\quad \text{for all $t\le T$.}
\end{equation}
The latter follows in particular if we can show that  $\Q$ is a martingale measure for the bond market specified by \eqref{TVMprior} and \eqref{Tbondprior}. We now derive necessary and sufficient conditions for this to hold. We thus let $h(t,T)$ be given by \eqref{TVMprior} and define the $T$-bond prices by \eqref{Tbondprior} and the short rates by \eqref{eqSR}. Here is our first result, which shows that the drift $\alpha(t,T)$ is fully determined by the requirement that $\Q$ is a local martingale measure.

\begin{proposition}
$\Q$ is a local martingale measure if and only if the following \emph{discount drift condition} holds
\begin{equation}\label{eqDDC}
  \alpha(t,T) = h(t,T) h(t,t)  .
\end{equation}
\end{proposition}

\begin{proof}
The implied $T$-bond price dynamics is
\begin{equation}\label{Pimpdyn}
  dP(t,T) =    -d\bigg(\int_t^T h(t,s)\,ds\bigg) = r_t\,dt  - \int_t^T \alpha(t,s)\,ds\,dt - \int_t^T \sigma(t,s)\,ds\,dW_t .
\end{equation}
On the other hand, the discounted $T$-price process $\e^{-\int_0^t r_s\, ds} P(t,T)$ is a $\Q$-local martingale if and only if the drift of $P(t,T)$ equals $P(t,T) r_t\,dt$. Matching this with the drift in \eqref{Pimpdyn} gives $\int_t^T \alpha(t,s)\,ds = H(t,T)r_t$. Differentiating in $T$, we obtain \eqref{eqDDC}, which proves the proposition.
\end{proof}

Given the drift condition \eqref{eqDDC} and using the existence of a local martingale measure as synonym for the absence of arbitrage, we can paraphrase that the generic dynamics \eqref{TVMprior} for an arbitrage-free discount derivative model is of the form
 \begin{equation}\label{TVMNA}
 \begin{aligned}
   dh(t,T) &= h(t,T) h(t,t)\,dt + \sigma(t,T)\,dW_t ,\quad 0\le t\le T,\\
   h(0,T) &= h_0(T),
\end{aligned}
 \end{equation}
for an initial discount derivative curve $h_0$. Note that, in contrast to the HJM drift condition on the forward rates, see, e.g., \cite[Theorem 6.1]{fil_09}, the derivative drift condition \eqref{eqDDC} does not depend on the volatility $\sigma(t,T)$.

We next show that, for a given volatility $\sigma(t,T)$, the system of stochastic differential equations~\eqref{TVMNA} uniquely determines $h(t,T)$.

\begin{lemma}\label{lemunique}
 For a given volatility process $\sigma(t,T)$, $0\le t\le T<\infty$ and initial discount derivative curve $h_0$, there exists at most one solution $h(t,T)$, $0\le t\le T<\infty$ to \eqref{TVMNA}.
\end{lemma}

\begin{proof}
 Let $h$ and $\tilde h$ be solutions to \eqref{TVMNA}, then their discounted difference
\[ q(t,T):=\e^{-\int_0^t \tilde h(s,s)\,ds}\big(\tilde h(t,T)-  h(t,T)\big)\]
satisfies
\begin{equation}\label{ODEdq}
  dq(t,T) =   h(t,T) q(t,t) \,dt,\quad q(0,T)=0.
\end{equation}
Integrating gives, for $T=t$,
\[ q(t,t) = \int_0^t h(s,t) q(s,s)\,ds,\quad \text{for all $t\ge 0$.} \]
We claim that $q(t,t)=0$ for all $t\ge 0$, and by \eqref{ODEdq} thus $\tilde h = h$. Indeed, by contradiction assume that there exist $0\le t_0<t_1$ such that $q(s,s)=0$, for all $s\le t_0$, and $|q(t_1,t_1)|>0$. We let $\beta(s)$ be a positive process such that $|h(s,t)|\le \beta(s)$ for all $t_0\le s\le t\le t_1$. Then $|q(t,t)| \le  \int_{t_0}^t \beta(s) |q(s,s)| \,ds$ for all $t\in [t_0,t_1]$, and Gr\"onwall's inequality implies that $q(t,t)=0$ for all $t\in [t_0,t_1]$. This contradicts the assumption, whence the claim is proved.
\end{proof}

The problem remains that we still have no guarantee that bond prices in \eqref{Tbondprior} are positive. Our main result is the following theorem, which provides sufficient conditions such that the discount framework \eqref{TVMNA} defines a feasible arbitrage-free price system for $T$-bonds. It thus represents an alternative to the HJM framework of forward rates.

\begin{theorem}\label{thmhtT}
Let $h(t,T)$, $0\le t\le T<\infty$, be any solution to \eqref{TVMNA} such that $r_t = h(t,t)$ is well defined and
  \begin{equation}\label{assthmmart}
   \text{$\int_0^t \e^{-\int_0^s r_u\,du}\sigma(s,T)\,dW_s$, $0\le t\le T$, is a $\Q$-martingale, for all $T$.}
  \end{equation}
Then $\Q$ is a martingale measure and the implied $T$-bond prices in \eqref{Tbondprior} satisfy
\begin{equation}\label{1HEQeq}
  P(t,T)  =  \E_\Q\bigg[ \e^{-\int_t^T r_s\,ds} \mid\Fcal_t\bigg] ,
\end{equation}
and are thus positive, that is, \eqref{inthle1} holds, in particular.
\end{theorem}

\begin{proof}
Let $\{ h(t,T) : 0\le t\le T<\infty\}$ be a solution to \eqref{TVMNA}. Denote $r_t = h(t,t)$ and set \[ M(t,T):=\e^{-\int_0^t r_s\,ds}h(t,T).\] Then $M(t,T)$ has zero drift and dynamics
\[ dM(t,T) = \e^{-\int_0^t r_s\,ds}\sigma(t,T)\,dW_t ,\]
and thus, by assumption \eqref{assthmmart}, is a $\Q$-martingale.  We obtain
\begin{equation}\label{htTeqX}
  h(t,T) = \e^{\int_0^t r_s\,ds} M(t,T) = \e^{\int_0^t r_s\,ds} \E_\Q[M(T,T)\mid \Fcal_t] = \E_\Q\Big[ \e^{-\int_t^T r_s\,ds}  r_T\mid\Fcal_t\Big] .
\end{equation}
Integrating over $T$ gives
\begin{equation} \label{HtTeqX}
H(t,T)=\int_t^T h(t,u)\,du = \E_\Q\bigg[ \int_t^T \e^{-\int_t^u r_s\,ds}  r_u\,du \mid\Fcal_t\bigg]=1- \E_\Q\Big[ \e^{-\int_t^T r_s\,ds}   \mid\Fcal_t\Big],
\end{equation}
as desired.
\end{proof}

The following remarks provide further details and discussion on the above discount framework.

\begin{remark}
The identities \eqref{htTeqX} and \eqref{HtTeqX} are of independent interest, and show the economic meaning of $h(t,T)$ and $H(t,T)$ as present values of future cash flows $r_T$ and $r_u$ for $u\in [t,T]$, respectively.
\end{remark}

\begin{remark}
The expression \eqref{Pimpdyn} shows that the induced volatility $v(t,T)$ of the $T$-bond returns is given by $P(t,T)v(t,T)=-\int_t^T \sigma(t,s)\,ds$.
\end{remark}

\begin{remark}
The equivalent physical measure $\Pa\sim\Q$ is related to $\Q$ by the market price of risk $\theta$ such that the Radon--Nikodym derivative satisfies $\E_\Q[\frac{d\Pa}{d\Q}\mid\Fcal_T] = \exp(\int_0^T \theta_t^\top dW_t -\frac{1}{2}\int_0^T \|\theta_t\|^2\,dt)$, for any time horizon $T>0$. This implies the $\Pa$-Brownian motion $dW^\Pa_t = dW_t -\theta_t \,dt$. Hence the dynamics under $\Pa$ of $h(t,T)$ is $dh(t,T) = (h(t,T)h(t,t) + \sigma(t,T)\theta_t )\,dt + \sigma(t,T)\, dW^\Pa_t $.
\end{remark}

\begin{remark}\label{remTP}
The proofs of Lemma~\ref{lemunique} and Theorem~\ref{thmhtT} rely on technical properties, e.g., that $T\mapsto h(t,T)$ is locally bounded in $T\ge t$ a.s., or that the order of integration and expectation can be changed. Such properties can be asserted by imposing sufficient technical assumptions.
\end{remark}

\begin{remark}\label{remquadrdrift}
Existence of a solution for the system of SDEs \eqref{TVMNA} is an open problem. Existence could be an issue in view of the quadratic drift \eqref{eqDDC}, which may cause explosion in finite time. Example \ref{DFMex1} and Section \ref{sswellbeh} below give feasible specifications. A natural approach for a systematic study is to state \eqref{TVMNA} as a stochastic partial differential equation for the discount derivative curve $\psi_t(x):=h(t,t+x)$ in an appropriate function space $\Hcal$. Such an SPDE is of the form
\begin{equation}\label{SPDEeq}
  d\psi_t(x)  = \big(\partial_x\psi_t(x) + \psi_t(x)\psi_t(0)\big)dt + B(\psi_t)\,dW_t ,
\end{equation}
for some appropriate volatility operator $B:\Hcal\to \Hcal^n$, so that $\sigma(t,t+\cdot)=B(\psi_t)$. An example of an appropriate function space is the weighted Sobolev space $\Hcal_w$ consisting of weakly differentiable functions $\psi:[0,\infty)\to\R$ with $\|\psi\|_w^2:=\int_0^\infty  \psi'(x)^2 w(x)\,dx<\infty$ and $\psi(\infty)=0$, for some increasing and continuously differentiable weight function $w$, such that $\int_0^\infty w^{-1/3}(x)\,dx<\infty$. E.g., $w(x)=\e^{\alpha x}$ for some $\alpha>0$. This is similar to the space $H_w$ introduced in \cite[Chapter 5]{fil_01}. It can be shown, as in \cite[Equation (5.7)]{fil_01}, that the $L^1$-norm of any $\psi\in\Hcal_w$ is bounded by $\int_0^\infty |\psi(x)|\,dx\le C_w \|\psi\|_w$, for some finite constant $C_w$. It can further be shown, as in \cite[Theorem 5.1.1]{fil_01}, that the differential operator $\partial_x$ generates a strongly continuous semigroup on $\Hcal_w$. Hence one can study existence and uniqueness of (local) mild and weak solutions to \eqref{SPDEeq} in the spirit of \cite{dap_zab_92}. In fact, uniqueness follows as soon as the volatility operator $B(\psi_t)$ is Lipschitz continuous in $\psi_t$, see \cite[Corollary 2.4.1]{fil_01}. This is a direct improvement of Lemma \ref{lemunique}, since there we assumed that the volatility process $\sigma(t,T)$ is given as exogenous and does not depend on $\psi_t$. Similarly, global existence would follow if one could show that any local weak solution of \eqref{SPDEeq} with $\|\psi_0\|_w< C_w^{-1}$ remains bounded, $\|\psi_t\|_w< C_w^{-1}$ for all $t\ge 0$. Combined with the above $L^1$-bound this would imply that $h(t,s)=\psi_t(s-t)$ satisfies \eqref{inthle1}, so that the induced bond prices \eqref{Tbondprior} are positive.

For the special, deterministic, case where $B=0$, the unique local solution to \eqref{SPDEeq} is given by
\[  \psi_t(x) = \frac{\psi_0(t+x)}{1-\int_0^t \psi_0(s)\,ds}.\]
This solution does not explode in finite time if and only if the initial curve satisfies \eqref{inthle1}, that is, $\int_0^t \psi_0(s)\,ds<1$ for all finite $t$. In this case, it follows easily that also $\psi_t$ satisfies \eqref{inthle1}, as $\int_0^{T-t} \psi_t(x)\,dx = (\int_0^T \psi_0(s)\,ds - \int_0^t \psi_0(s)\,ds)(1-\int_0^t \psi_0(s)\,ds)^{-1} < 1$ for all finite $T\ge t$.
\end{remark}

The following example illustrates that the discount framework \eqref{TVMNA} admits feasible solutions.

\begin{example}\label{DFMex1}
A simple toy model specification is $h(t,T)=\phi(T-t) r_t$, for some deterministic function $\phi$ with $\phi(0)=1$, and where the short rate $r_t$ follows an It\^o process of the form $dr_t = \mu_t\,dt + \nu_t\,dW_t$. The induced dynamics of $h(t,T)$ is
\[ dh(t,T) = (\phi(T-t)\mu_t -\phi'(T-t) r_t)\,dt + \phi(T-t)\nu_t\,dW_t.\]
The drift condition \eqref{eqDDC} now reads as consistency condition
\begin{equation}\label{eqXX1}
 \phi(T-t)\mu_t -\phi'(T-t) r_t = \phi(T-t)r_t^2,
\end{equation}
which does not depend on the volatility process $\nu_t$, whereas the volatility in \eqref{TVMNA} is simply induced as $\sigma(t,T)= \phi(T-t)\nu_t$. This property holds more general in affine discount term structure models, see Remark~\ref{remSPAN} below.

For $T=t$, as $\phi(0)=1$, we obtain the short rate drift $\mu_t = (r_t+\phi'(0))r_t $ and volatility $\nu_t=\sigma(t,t)$. Assume that $\phi'(0)=-\theta$ is negative, for some parameter $\theta>0$, and assume that $\sigma(t,t)\to 0$ fast enough as $r_t\to 0$ and $r_t\to \theta$, respectively, then the risk-neutral dynamics
\[ dr_t = -(\theta-r_t)r_t\,dt + \sigma(t,t)\,dW_t \]
is well-behaved with values in $[0,\theta]$, for any initial value $r_0\in [0,\theta]$. This short rate dynamics is reminiscent of the linear-rational framework \cite[Section IV.E]{fil_lar_tro_17}.

Plugging $\mu_t=-(\theta-r_t)r_t$ back in \eqref{eqXX1}, we obtain $-\phi(T-t)\theta r_t -\phi'(T-t)r_t=0$, which is equivalent to $\phi(s)=\e^{-\theta s}$. This results in a term structure of $T$-discounts of
\[ H(t,T) = r_t \int_0^{T-t}\phi(s)\,ds =  \big(1- \e^{-(T-t)\theta}\big)\frac{r_t}{\theta}\]
and thus a term structure of $T$-bonds of
\[ P(t,T) = 1 - \big(1- \e^{-(T-t)\theta}\big)\frac{r_t}{\theta} .\]
As a technical note, from above we have $r_t$ takes values in $[0,\theta]$, which implies that the discount curve $T\mapsto P(t,T)$ is decreasing and $P(t,T)>0$ for all finite maturities $T>t$.
\end{example}

\section{Discount factor models}

We now elaborate on discount factor models, extending Example \ref{DFMex1}. We first derive the consistency conditions for a general factor model. We then study in more detail the affine discount factor models.

\subsection{General factor models}
We study consistent discount factor models of the form
  \[ h(t,T) = \phi(T-t,Z_t) \]
  for some function $\phi:[0,\infty)\times \Zcal\to \R$, where $\Zcal\subseteq\R^d$ is some state space with non-empty interior, and the factor $Z_t$ is some $\Zcal$-valued diffusion process, with dynamics
\[ dZ_t = \mu(Z_t)\,dt + \nu(Z_t)\,dW_t ,\]
for a drift function $\mu:\Zcal\to \R^d$ and volatility function $\nu:\Zcal\to \R^{d\times n}$. We denote the corresponding diffusion function by $c(\cdot)=\nu(\cdot)^\top \nu(\cdot)$.
The induced dynamics of $h(t,T)$ is
\begin{align*}
  dh(t,T) &=\Big(-\partial_1 \phi(T-t,Z_t) +\mu(Z_t)^\top\nabla_z\phi(T-t,Z_t) +\frac{1}{2}\trace[c(Z_t)\nabla_z^2 \phi(T-t,Z_t)]\Big)\,dt \\
  &\quad + \nabla_z\phi(T-t,Z_t)^\top\nu(Z_t)\,dW_t.
\end{align*}
Matching the drift term with the arbitrage-free dynamics \eqref{TVMNA} point-wise gives the following consistency equation
\begin{equation}
 -\partial_x \phi(x,z) +\mu(z)^\top\nabla_z\phi(x,z) +\frac{1}{2}\trace[c(z)\nabla_z^2 \phi(x,z)] = \phi(x,z)\phi(0,z) ,\label{conseq}
\end{equation}
whereas the induced volatility is
\begin{equation}\label{dervolindu}
 \sigma(t,T)= \nabla_z\phi(T-t,Z_t)^\top\nu(Z_t).
\end{equation}

\subsection{Affine discount term structure models}
We now assume an affine term structure
\begin{equation}\label{ATS}
\phi(x,z) = \phi_0(x) + \sum_{i=1}^d \phi_i(x) z_i ,
\end{equation}
for some functions $\phi_j:[0,\infty)\to \R$, $j=0,1,\dots,d$. Plugging in the affine term structure \eqref{ATS} in the consistency equation \eqref{conseq} gives
\begin{equation}\label{ATScons}
   -\phi_0'(x) -\sum_{i=1}^d \phi_i'(x)z_i + \sum_{i=1}^d \phi_i(x)\mu_i(z) = \phi_0(x)\gamma_0 +\sum_{j=1}^d \big(\phi_0(x)\gamma_j + \phi_j(x)\gamma_0\big) z_j + \sum_{j,k=1}^d \phi_j(x)\gamma_k z_j z_k,
\end{equation}
where we denote
\[  \gamma_j := \phi_j(0),\quad j=0,1,\dots,d.\]

\begin{remark}\label{remSPAN}
 We note that the consistency condition \eqref{ATScons} does not depend on the diffusion matrix $c(z)$ of the factor process $Z_t$. This is in contrast to the non-linear case as seen in \eqref{conseq} where $c(z)$ shows up. In other words, an affine discount term structure leaves the underlying volatility unspanned. This is in contrast to affine models of the forward rates, see \cite{fil_tro_sta_19}, and reminiscent of the linear-rational framework, see \cite[Section I]{fil_lar_tro_17}. In fact, it follows directly from~\eqref{ATS} that $T$-bond prices become affine expressions in $Z_t$. We derive the explicit expressions in~\eqref{bondlinearexp} below.
\end{remark}

We henceforth assume that the functions
\begin{equation}\label{linass}
  \text{$\phi_1,\dots,\phi_d$ are linearly independent.}
\end{equation}
Equation \eqref{ATScons} then implies that every drift function $\mu_i(z)$ is a quadratic polynomial in $z$,
\begin{equation}\label{driftZquadprior}
 \mu_i(z) = b_i + \sum_{j=1}^d \beta_{ij}z_j + \sum_{j,k=1}^d B_{i,jk} z_j z_k
\end{equation}
for some coefficients $b_i$, $\beta_{ij}$, $B_{i,jk}$. Plugging \eqref{driftZquadprior} in \eqref{ATScons} and matching coefficients of same order in $z$ gives
\begin{align}
  -\phi_0'(x) + \sum_{i=1}^d b_i \phi_i(x) &= \phi_0(x)\gamma_0 ,\label{Acc1}\\
  -\phi_j'(x) + \sum_{i=1}^d \beta_{ij} \phi_i(x) &= \phi_0(x)\gamma_j + \phi_j(x)\gamma_0,\quad j=1,\dots,d, \label{Acc2}\\
   \sum_{i=1}^d (B_{i,jk}+B_{i,kj}) \phi_i(x)&=  \phi_j(x)\gamma_k+\phi_k(x)\gamma_j ,\quad j,k=1,\dots,d. \label{Acc3}
\end{align}

We arrive at the following result, where we write $\bar Z_t = (1,Z_{1,t},\dots,Z_{d,t})^\top$ for the extended factor process including the constant 1, $e_0=(1,0,\dots,0)^\top$, $\bar\phi(x) := (\phi_0(x),\dots,\phi_d(x))^\top$, and $\bar\gamma:=(\gamma_0,\dots,\gamma_d)^\top =\bar\phi(0)$.

\begin{proposition}
Assume \eqref{linass}. The factor process $Z_t$ in an arbitrage-free affine discount term structure model of the form \eqref{ATS} has a quadratic drift of the form
\begin{equation}\label{driftZquadcons}
 \mu_i(z) = b_i + \sum_{j=1}^d \beta_{ij}z_j +  z_i \sum_{j=1}^d  \gamma_j  z_j,\quad i=1,\dots, d,
\end{equation}
for some coefficients $b_i$, $\beta_{ij}$, and $\gamma_j$. The functions $\phi_i$ in turn are given by
\begin{equation}\label{phisol}
   \bar\phi(x) = \e^{Ax}\bar\gamma ,
\end{equation}
where $A$ is defined as the $(d+1)\times (d+1)$ matrix on the right hand side of \eqref{ODEphifull} below. The $T$-discounts and short rates are linear in $\bar Z_t$,
\begin{align}
 H(t,T) &=  1- e_0^\top \e^{A^\top (T-t)} \bar Z_t,\label{TdiscountATS}\\
 r_t &= \bar\gamma^\top \bar Z_t.\label{srATS}
\end{align}

\end{proposition}

\begin{proof}
In view of \eqref{linass}, condition \eqref{Acc3} is equivalent to
\begin{equation*}
  B_{i,jk}   = \begin{cases}
    \gamma_k,&\text{if $i=j$,}\\
    \gamma_j,&\text{if $i=k$,}\\
    0,&\text{otherwise.}
  \end{cases}
\end{equation*}
Hence the drift \eqref{driftZquadprior} takes the form \eqref{driftZquadcons}.

Equations \eqref{Acc1}--\eqref{Acc2} are equivalent to the linear system of ordinary differential equations
\begin{equation}\label{ODEphifull}
  \begin{pmatrix}
  \phi_0'(x) \\ \phi_1'(x) \\ \vdots \\ \phi_i'(x) \\ \vdots \\ \phi_d'(x)
\end{pmatrix}
=\underbrace{\begin{pmatrix}
  -\gamma_0 & b_1 & b_2 & \dots  & b_{d-1}  & b_d \\
   -\gamma_1 & \beta_{11} -\gamma_0 & \beta_{21} & \dots & \beta_{d-1,1} & \beta_{d1} \\
 \vdots & \vdots & \ddots &  & & \vdots \\
   -\gamma_i & \beta_{1i}   &  & \beta_{ii} -\gamma_0 &  & \beta_{di} \\
 \vdots & \vdots &  & & \ddots & \vdots \\
  -\gamma_d & \beta_{1d}   & \beta_{2d} & \dots & \beta_{d-1,d} & \beta_{dd} -\gamma_0
  \end{pmatrix}}_{=:A}
  \begin{pmatrix}
  \phi_0(x) \\ \phi_1(x) \\ \vdots \\ \phi_i(x) \\ \vdots \\ \phi_d(x)
\end{pmatrix}.
\end{equation}
The system of ODEs \eqref{ODEphifull} can be written in compact form as
\begin{equation}\label{ODEphi}
\begin{aligned}
  \bar{\phi}'(x) &= A \bar\phi(x),\\
  \bar\phi(0)&= \bar\gamma.
\end{aligned}
\end{equation}
The solution of \eqref{ODEphi} is \eqref{phisol}.

Note that
\begin{equation}\label{eqbargamma}
 \bar\gamma = -A e_0.
\end{equation}
Hence $\bar\phi(x) =  - \e^{Ax}A e_0 =-\partial_x \e^{Ax} e_0$, and the primitive functions
\[ \Phi_j(x) = \int_0^x\phi_j(u)\,du ,\quad j=0,\dots,d,\]
are given by
\[ \bar\Phi(x)=(\Phi_0(x),\dots,\Phi_d(x))^\top = \big( I_{d+1} - \e^{Ax}   \big)e_0.\]

The expressions \eqref{TdiscountATS} and \eqref{srATS} for the $T$ discount $H(t,T) = \bar\Phi(T-t)^\top\bar Z_t$ and short rates $r_t$ now follow by elementary calculus.
\end{proof}

As announced in Remark \ref{remSPAN}, we conclude from \eqref{TdiscountATS} that the induced $T$-bond prices are linear in $\bar Z_t$,
\begin{equation}\label{bondlinearexp}
  P(t,T) =   e_0^\top \e^{A^\top (T-t)} \bar Z_t.
\end{equation}
Straightforward differentiation using \eqref{eqbargamma} shows that forward rates are linear-rational in $\bar Z_t$,
\begin{equation*}
  f(t,T) =   \frac{\bar\gamma^\top \e^{A^\top (T-t)} \bar Z_t }{ e_0^\top \e^{A^\top (T-t)} \bar Z_t}.
\end{equation*}

\subsection{Well behaved factor processes}\label{sswellbeh}

As noted in Remark~\ref{remquadrdrift}, existence of global factor processes may be an issue in view of their quadratic drift \eqref{driftZquadcons}, which may cause explosion in finite time. We discuss here some well behaved specifications.

More specifically, we show that there exists factor processes $Z_t$, with quadratic drift of the form \eqref{driftZquadcons}, taking values in the half-open solid simplex
\begin{equation}\label{Zcalass}
  \Zcal:=\bigg\{ z\in [0,1] : \sum_{i=1}^d z_i < 1\bigg\}.
\end{equation}
This can always be achieved, as we show now. Thereto, we first specify a diffusion $U_t$ with values in $[0,\infty)^d$ and then apply the diffeomorphism $G:[0,\infty)^d\to \Zcal$ defined by
\[ G_i(u) := u_i \bigg(1+\sum_{j=1}^d u_j\bigg)^{-1},\quad i=1,\dots, d, \]
with inverse $G^{-1}_i (z) = z_i   (1-\sum_{j=1}^d z_j  )^{-1}$.

The dynamics of $U_t$ could be of the form
\begin{equation}\label{Udyn}
  dU_{t,i} = \bigg( \kappa_i U_{t,i} + \theta_i\bigg(1+\sum_{j=1}^d U_{t,j}\bigg)\bigg)dt + q_i\sqrt{U_{t,i}\bigg(1+\sum_{j=1}^d U_{t,j}\bigg)}dW_{t,i},
\end{equation}
for some parameters $\kappa_i\in\R$, $\theta_i,q_i\ge 0$.

\begin{lemma}\label{lemUZdyn}
There exists a $[0,\infty)^d$-valued weak solution $U_t$ to \eqref{Udyn}. The drift of the transformed process $Z_{t}:=G(U_t)$  is quadratic in $Z_t$ of the form \eqref{driftZquadcons}, and assumption \eqref{assthmmart} holds, so that Theorem~\ref{thmhtT} applies.
\end{lemma}

\begin{proof}
Existence of a $[0,\infty)^d$-valued weak solution $U_t$ to \eqref{Udyn} follows from \cite[Theorem 5.3 and Proposition 6.4]{fil_lar_16}.

Now define $V_t:= 1+\sum_{j=1}^d U_{t,j}$, so that $Z_{t,i}=U_{t,i}(V_t)^{-1}$, and denote $\theta_V:=\sum_{j=1}^d \theta_j$. Using It\^o calculus, we obtain
 \begin{align*}
   dU_{t,i} & = \big( \kappa_i U_{t,i} + \theta_i V_t\big)\, dt + q_i\sqrt{U_{t,i}} \sqrt{V_t}\,dW_{t,i},
 \end{align*}
and thus
 \begin{align*}
   dV_t & = \sum_{j=1}^d dU_{t,j} = \bigg( \sum_{j=1}^d \kappa_j U_{t,j} + \theta_V V_t\bigg) dt + \sum_{j=1}^d q_j\sqrt{U_{t,j}} \sqrt{V_t}\,dW_{t,j},
 \end{align*}
 and
 \begin{align*}
   d(V_t)^{-1} &= -(V_t)^{-2} dV_t + (V_t)^{-3} d\langle V,V\rangle_t \\
    & =-(V_t)^{-1}\bigg( \sum_{j=1}^d \kappa_j Z_{t,j} + \theta_V  \bigg) dt -(V_t)^{-1} \sum_{j=1}^d q_j\sqrt{Z_{t,j}} \,dW_{t,j} + (V_t)^{-1}\sum_{j=1}^d q_j^2  Z_{t,j}\,dt.
 \end{align*}
Integration by parts thus gives
\begin{equation}\label{SDEZexi}
 \begin{aligned}
  dZ_{t,i} & = U_{t,i}\,d(V_t)^{-1} + (V_t)^{-1}\,dU_{t,i} + d\langle U_{i},V^{-1}\rangle_t \\
   & = -Z_{t,i}   \bigg( \sum_{j=1}^d \kappa_j Z_{t,j} + \theta_V  \bigg) dt -Z_{t,i} \sum_{j=1}^d q_j\sqrt{Z_{t,j}} \,dW_{t,j} + Z_{t,i}   \sum_{j=1}^d q_j^2 Z_{t,j}  \,dt \\
   &\quad + \big( \kappa_i Z_{t,i} + \theta_i \big)\, dt +  q_i\sqrt{Z_{t,i}}  \,dW_{t,i}  + q_i^2 Z_{t,i}\,dt \\
   &= \bigg( \theta_i + \big(-\theta_V+\kappa_i+q_i^2    \big)Z_{t,i} + Z_{t,i}\sum_{j=1}^d \big(-\kappa_j+q_j^2\big) Z_{t,j}  \bigg) dt \\
   &\quad + q_i \big(1-Z_{t,i} \big)\sqrt{Z_{t,i}}\,dW_{t,i} - Z_{t,i}\sum_{j\neq i}  q_j\sqrt{Z_{t,j}} \,dW_{t,j}.
\end{aligned}
\end{equation}
This shows that the drift function of $Z_t$ is of the form \eqref{driftZquadcons}.

The induced short rate process $r_t$ is linear in $Z_t$ given by \eqref{srATS} and thus uniformly bounded. The volatility function of $Z_t$ is also uniformly bounded, and hence the induced discount derivative volatility given in \eqref{dervolindu}. Hence assumption \eqref{assthmmart} holds, which completes the proof.
 \end{proof}

The existence result in Lemma \ref{lemUZdyn} is only partially an answer to the question about the global existence of SPDE \eqref{SPDEeq} in Remark~\ref{remquadrdrift}. In fact, the affine discount term structure model \eqref{ATS} generates only discount derivative curves of the form $\psi_t(x) = \phi_0(x) + \sum_{i=1}^d \phi_i(x) Z_{t,i}$, which lie in a $d$-dimensional affine subspace, say $\Acal$, of $\Hcal$. In view of \eqref{linass}, there exist points $0\le x_1<x_2<\cdots<x_d$ such that the factor $Z_t$ is an affine function of the $d$ discount derivative values $(\psi_t(x_1),\dots,\psi_t(x_d))$. Combining this with \eqref{dervolindu} and \eqref{SDEZexi}, we easily see that the volatility operator $B(\psi_t)$ in \eqref{SPDEeq} is a (non-Lipschitz continuous) function of the values $(\psi_t(x_1),\dots,\psi_t(x_d))$ when restricted to $\psi_t\in\Acal$. However, it is not clear how to extend $B$ beyond $\Acal$ so a solution to \eqref{SPDEeq} still exists outside $\Acal$. And uniqueness remains a problem anyway because of the lack of Lipschitz continuity at $Z_{t,i}=0$ of the volatility of $Z_t$ given in \eqref{SDEZexi}, as the following remark shows.

\begin{remark}\label{remunique}
While pathwise uniqueness for \eqref{Udyn} is proved in \cite[Theorem 4.3]{fil_lar_16} in dimension $d=1$, it remains an open problem in higher dimensions $d>1$. I conjecture that pathwise uniqueness for \eqref{Udyn} can be proved along similar arguments as used in \cite{yam_wat_71}. However, as discussed in \cite[Section 4]{fil_lar_16}, the problem is far from trivial. E.g., the results in \cite{yam_wat_71} do not directly apply here, as they assume that the $i$th element of the diagonal diffusion matrix only depends on the $i$th coordinate of the process.
\end{remark}

\section{Conclusion}
An arbitrage-free dynamic discount model defines an arbitrage-free price system for bonds. Modeling discount derivatives thus provides a valuable alternative to modeling forward rates. Of particular interest are affine discount term structure models, for which I provide feasible specifications.

The paper identifies various open problems outlined in Remarks \ref{remTP}, \ref{remquadrdrift}, and \ref{remunique}, which point to directions for further research. Other research directions include the implementation of discount models for pricing and hedging interest rate derivatives.

\bibliographystyle{apalike}

\bibliography{bib}

\end{document}